\documentclass{article}
\usepackage{ifthen}
\usepackage{mdwlist}
\usepackage{amsmath,amssymb,amsfonts,amsthm}
\usepackage{bm}
\usepackage[colorlinks=true,linkcolor=blue,citecolor=blue,urlcolor=blue]{hyperref}
\usepackage{enumitem}
\usepackage{graphicx}
\usepackage{xspace}
\usepackage{verbatim}
\usepackage{algorithm}
\usepackage{algpseudocode}
\usepackage[margin=1in]{geometry}
\usepackage{color}
\usepackage{thm-restate}
\usepackage{latexsym}
\usepackage{epsfig}
\usepackage{changepage}
\usepackage{multirow}
\def\d{\delta}

\def\eps{\ve}
\renewcommand{\epsilon}{\ve}
\def\ve{\varepsilon}

\newcommand{\E}{\mbox{\bf E}}

\newcommand{\pr}[2][]{\mbox{Pr}\ifthenelse{\not\equal{}{#1}}{_{#1}}{}\!\left[#2\right]}

\newcommand{\dtv}{d_{\mathrm {TV}}}
\newcommand{\dk}{d_{\mathrm K}}

\newtheorem{theorem}{Theorem}

\newtheorem{proposition}{Proposition}

\newtheorem{fact}{Fact}
\newtheorem{lemma}{Lemma}

\newtheorem{corollary}{Corollary}

\newtheorem{definition}{Definition}
\newtheorem{question}{Question}

\newcommand{\ignore}[1]{}
\providecommand{\poly}{\operatorname*{poly}}

\newenvironment{prevproof}[2]{\noindent {\em {Proof of {#1}~\ref{#2}:}}}{$\hfill\qed$\vskip \belowdisplayskip}

\newcommand{\bg}[1]{\medskip\noindent{\bf #1}}
\definecolor{Red}{rgb}{1,0,0}
\def\red{\color{Red}}

\newcommand{\oldbound}[1]{{}}

\newcommand{\SAMP}{{\sf SAMP}\xspace}

\newcommand{\Bin}{B}
\newcommand{\NACOND}{{\sf NACOND}\xspace}
\newcommand{\COND}{{\sf COND}\xspace}
\newcommand{\EVAL}{{\sf EVAL}\xspace}
\newcommand{\ANACONDA}{\textsc{Anaconda}}

\title{\ANACONDA: A Non-Adaptive Conditional Sampling Algorithm for Distribution Testing}

\author {
Gautam Kamath\thanks{Supported by NSF Award CCF-1551875, CCF-1617730, CCF-1650733, CCF-1741137, ONR N00014-12-1-0999, a Google Faculty Research Award, and a Simons Investigator award. This work was done primarily while the author was an student at MIT, and partially while he was an intern at Microsoft Research, New England.} \\
Simons Institute for the Theory of Computing\\
\tt{g@csail.mit.edu}
\and
Christos Tzamos\thanks{This work was done while the author was a postdoc at Microsoft Research, New England.} \\
University of Wisconsin-Madison \\
\tt{tzamos@wisc.edu}
}
\begin{document}
\maketitle
\begin{abstract}
We investigate distribution testing with access to non-adaptive conditional samples.
In the conditional sampling model, the algorithm is given the following access to a distribution: it submits a query set $S$ to an oracle, which returns a sample from the distribution conditioned on being from $S$.
In the non-adaptive setting, all query sets must be specified in advance of viewing the outcomes.

Our main result is the first polylogarithmic-query algorithm for equivalence testing, deciding whether two unknown distributions are equal to or far from each other.
This is an exponential improvement over the previous best upper bound, and demonstrates that the complexity of the problem in this model is intermediate to the the complexity of the problem in the standard sampling model and the adaptive conditional sampling model.
We also significantly improve the sample complexity for the easier problems of uniformity and identity testing.
For the former, our algorithm requires only $\tilde O(\log n)$ queries, matching the information-theoretic lower bound up to a $O(\log \log n)$-factor.

Our algorithm works by reducing the problem from $\ell_1$-testing to $\ell_\infty$-testing, which enjoys a much cheaper sample complexity.
Necessitated by the limited power of the non-adaptive model, our algorithm is very simple to state.
However, there are significant challenges in the analysis, due to the complex structure of how two arbitrary distributions may differ.
\end{abstract}

\newpage
\section{Introduction}
\label{sec:intro}
Statistical hypothesis testing is one of the most classical problems in statistics, with a rich history over the past century.
Over the last two decades, the problem has recently attracted the focus of theoretical computer scientists, primarily with a focus on rigorous, finite sample guarantees for distributions with large domain sizes.
The seminal works of Goldreich, Goldwasser, and Ron~\cite{GoldreichGR96}, Goldreich and Ron~\cite{GoldreichR00}, and Batu, Fortnow, Rubinfeld, Smith, and White~\cite{BatuFRSW00} initiated this study of distribution testing, viewing distributions as a natural domain for property testing (see~\cite{Goldreich17} for coverage of this much broader field).

Since these works, distribution testing has enjoyed a wealth of study, resulting in a thorough understanding of the complexity of testing many distributional properties (see e.g.~\cite{BatuFFKRW01, BatuKR04, Paninski08, AcharyaDJOP11, BhattacharyyaFRV11, Valiant11, Rubinfeld12, IndykLR12, DaskalakisDSVV13, ChanDVV14, ValiantV17a, Waggoner15, AcharyaD15, AcharyaDK15, BhattacharyaV15, DiakonikolasKN15a, DiakonikolasK16, Canonne16, BlaisCG17, BatuC17, DaskalakisDK18, DaskalakisKW18, DiakonikolasGPP18}, and~\cite{Rubinfeld12,Canonne15a, BalakrishnanW17b, Kamath18} for recent surveys).
For many problems, these works have culminated in sample-optimal algorithms.

In this paper, we will be concerned with the following three problems, defined on discrete distributions over $[n]$:
\begin{itemize}
\item {\bf Uniformity Testing:} Given sample access to a distribution $p$, test whether $p = \mathcal{U}_n$ (the uniform distribution on $[n]$) or is far from it;
\item {\bf Identity Testing:} Given sample access to a distribution $p$ and the description of a distribution $q$, test whether $p = q$ or is far from it;
\item {\bf Equivalence Testing:} Given sample access to distributions $p$ and $q$, test whether they are equal to or far from each other.
\end{itemize}
Observe that each of these problems generalizes the previous, and thus are in increasing difficulty.
All three of these problems have a sample complexity which is either $\Theta(n^{1/2})$ or $\Theta(n^{2/3})$.
In other words, while these problems enjoy a sample complexity which is strongly sublinear in the domain size, in the absence of additional assumptions, information-theoretic lower bounds often necessitate a sample complexity which is polynomial in the size of the domain.
When the domain is exceptionally large, this cost may be prohibitive for many of the inference tasks we wish to perform.

To circumvent these strong lower bounds, one may imagine oracle models where one has additional power when interacting with the distribution.
Some examples include when the algorithm may query the PDF or CDF of the distribution~\cite{BatuDKR05, GuhaMV06, RubinfeldS09, CanonneR14}, or is given probability-revealing sample~\cite{OnakS18}.
However, perhaps the most popular alternative model, and the one we consider in this paper, is the \emph{conditional sampling} model.
This model was recently introduced concurrently by Chakraborty, Fischer, Goldhirsh, and Matsliah~\cite{ChakrabortyFGM13, ChakrabortyFGM16} and Canonne, Ron, and Servedio~\cite{CanonneRS14, CanonneRS15}.
The algorithm is able to \emph{query} a distribution in the following way: it submits a query set $S$ to an oracle, which returns a sample from the distribution conditioned on being from $S$.
Additionally, we will distinguish between conditional sampling models where the algorithm's queries may be adaptive (\COND) or non-adaptive (\NACOND).\footnote{A formal definition of these concepts is given in Definition~\ref{def:cond}.}
In comparison, we will use \SAMP to refer to the standard sampling model.

Conditional sampling often dramatically reduces the complexity of distribution testing problems.
For example, given \SAMP access to a distribution, the sample complexity of identity testing is $\Theta(\sqrt{n}/\ve^2)$~\cite{Paninski08,ValiantV17a}.
However, given \COND access, the query complexity drops to $\tilde \Theta(1/\ve^2)$~\cite{FalahatgarJOPS15}, completely removing the dependence on the support size.
Motivated by the power of this model, there has been significant investigation into its implications on distribution testing~\cite{Canonne15b, FalahatgarJOPS15, AcharyaCK15b, FischerLV17, SardharwallaSJ17, BlaisCG17, BhattacharyyaC18}, as well as group testing~\cite{AcharyaCK15a}, sublinear algorithms~\cite{GouleakisTZ17}, and crowdsourcing~\cite{GouleakisTZ18}.

At this point, we have a developed understanding of the power of the \COND oracle with respect to the aforementioned distribution testing problems.
Perhaps surprisingly, the relative complexities of certain problems have qualitatively different relationships between \SAMP and \COND.
To be precise, the sample complexities of identity testing and equivalence testing in \SAMP are $\Theta(n^{1/2})$~\cite{Paninski08, ValiantV17a} and $\Theta(n^{2/3})$~\cite{ChanDVV14} respectively: there is a polynomial relationship between the two.
However, their query complexities in \COND are $\Theta(1)$~\cite{CanonneRS15,FalahatgarJOPS15} and $\log^{\Theta(1)}\log n$~\cite{FalahatgarJOPS15,AcharyaCK15b} respectively: there is a ``chasm'' between the two complexities, as we go from no dependence on the domain size to a doubly logarithmic one.

However, the picture is much less clear when it comes to the non-adaptive \NACOND model.
We know that the complexity of identity testing is $\poly \log n$~\cite{ChakrabortyFGM13,AcharyaCK15b}, though the upper and lower bounds are quite far from each other.
On the other hand, the complexity of equivalence testing is far less clear: the best lower bound is $\Omega(\log n)$~\cite{AcharyaCK15b}, and the best upper bound is $O(n^{2/3})$~\cite{ChanDVV14}.
Given the interesting qualitative behavior observed for the \COND model, this begs the following question:
\begin{question}
What is the relationship of the query complexities of identity and equivalence testing in the \NACOND model?
\end{question}
In particular, are they polynomially related, as in the \SAMP model?
Or is there a larger gap between the two, as in the \COND model?
Stated another way, do we require both conditional samples and adaptivity \emph{simultaneously} in order to reap the benefits for testing equivalence?

\subsection{Results and Discussion}

Our main result is a qualitative resolution to this problem: we give a $\poly \log n$-query algorithm for equivalence testing.
\begin{theorem}[Non-Adaptive Equivalence Testing]
\label{thm:equivalence}
There exists an algorithm which, given \NACOND access to unknown distributions $p, q$ on $[n]$, makes $\tilde O\left(\frac{\log^{12} n}{\ve^2}\right)$ queries to the oracle on each distribution and distinguishes between the cases $p = q$ and $\dtv(p,q) \geq \ve$ with probability at least $2/3$.  
\end{theorem}

For the special case of uniformity testing, we have a sharper analysis, allowing us to obtain a $\tilde O(\log n)$ query algorithm, which nearly matches the $\Omega(\log n)$ lower bound of~\cite{AcharyaCK15b}:
\begin{theorem}[Non-Adaptive Uniformity Testing]
\label{thm:uniform}
There exists an algorithm which, given \NACOND access to an unknown distribution $p$ on $[n]$, makes $\tilde O\left(\frac{\log n}{\ve^2}\right)$ queries to the oracle on $p$ and distinguishes between the cases $p = \mathcal{U}_n$ and $\dtv(p,\mathcal{U}_n) \geq \ve$ with probability at least $2/3$, where $\mathcal{U}_n$ is the uniform distribution on $[n]$.
\end{theorem}

As a corollary of Theorem~\ref{thm:uniform}, we can obtain an improved upper bound for identity testing with an adaptation of the reduction from identity testing to uniformity testing of~\cite{ChakrabortyFGM16} (inspired by the bucketing techniques of~\cite{BatuFRSW00,BatuFFKRW01}).
\begin{theorem}[Non-Adaptive Identity Testing]
\label{thm:identity}
There exists an algorithm which, given \NACOND access to an unknown distribution $p$ on $[n]$ and a description of a distribution $q$ over $[n]$, makes $\tilde O\left(\frac{\log^2 n}{\ve^2}\right)$ queries to the oracle on $p$ and distinguishes between the cases $p = q$ and $\dtv(p,q) \geq \ve$ with probability at least $2/3$.
\end{theorem}

\noindent Our results and a comparison with the complexity of testing in various oracle models are presented in Table~\ref{fig:batable}.
\newcommand{\pb}[2]{\parbox[c][][c]{#1}{\strut#2\strut}}
\begin{table}[ht]\centering
      \begin{tabular}{| c | c | c | c |}
      \hline
      {\bf Model} & {\bf Uniformity} & {\bf Identity} & {\bf Equivalence} \\ \hline
      {\bf \SAMP} & $\Theta\left(\frac{\sqrt{n}}{\ve^2}\right) \cite{Paninski08, ValiantV17a}$ & $\Theta\left(\frac{\sqrt{n}}{\ve^2}\right) \cite{Paninski08, ValiantV17a}$ & $\Theta\left(\max\left(\frac{n^{2/3}}{\eps^{4/3}}, \frac{n^{1/2}}{\eps^2}\right) \right)$ \cite{ChanDVV14}\\  \hline
      \multirow{2}{*}{\bf \centering \NACOND} & $\tilde O\left(\frac{\log n}{\ve^2}\right) \textsf{\red [this work]}$& 
        $\tilde O\left(\frac{\log^{2} n}{\ve^2}\right) \textsf{\red [this work]}$ & $\tilde O\left(\frac{\log^{12} n}{\ve^2}\right) \textsf{\red [this work]}$\\
      & $\Omega\left(\log n\right)$ \cite{AcharyaCK15b}  & $\Omega\left(\log n\right) \cite{AcharyaCK15b}$ & $\Omega\left(\log n\right) \cite{AcharyaCK15b}$ \\ \hline
      \multirow{2}{*}{{\bf \centering \COND}} & $\tilde O\left(\frac{1}{\ve^2}\right) \cite{CanonneRS15}$& $\tilde O\left(\frac{1}{\ve^2}\right) \cite{FalahatgarJOPS15}$& 
        $\tilde O\left(\frac{\log \log n}{\ve^5}\right) \cite{FalahatgarJOPS15}$ \\
      & $\Omega\left(\frac{1}{\ve^2}\right)$\cite{CanonneRS15} & $\Omega\left(\frac{1}{\ve^2}\right)$\cite{CanonneRS15} & $\Omega\left(\sqrt{\log \log n}\right) \cite{AcharyaCK15b}$\\ \hline
      \end{tabular}
    \caption{Summary of results, and a comparison of uniformity, identity, and equivalence testing in different sampling oracle models. Problems get harder as one moves up and to the right in this table.}
\label{fig:batable}
\end{table}

\iffalse
\begin{table}[ht]\centering
    \begin{adjustwidth}{-.1in}{-.1in}\centering
      \begin{tabular}{| c | c | c |}
      \hline
      \pb{0.30\textwidth}{\centerline{\bf Model}} & {\bf Identity Testing} & {\bf Equivalence Testing} \\ \hline
      {\bf  \EVAL} & $ \Theta\left(\frac{1}{\ve}\right) \cite{RubinfeldS09}$ & $\Theta(n)$\\ \hline
      \multirow{2}{*}{\pb{0.30\textwidth}{\bf \centering PR-Samples}} & \multirow{2}{*}{$\Theta\left(\frac{1}{\ve}\right) \cite{RubinfeldS09}$}& 
        $O\left(\max\left(\frac{\sqrt{n}}{\ve}, \frac{1}{\ve^2}\right)\right) \cite{OnakS16}$ \\
      &   & $\Omega\left(\frac{\sqrt{n}}{\ve}\right) \cite{OnakS16}$  \\ \hline
      {\bf  Dual, Cumulative Dual} & $\Theta\left(\frac{1}{\ve}\right) \cite{CanonneR14}$ & $\Theta\left(\frac{1}{\ve}\right) \cite{CanonneR14}$\\ \hline
      \end{tabular}
    \end{adjustwidth}
    \caption{More oracle models.}
\label{fig:batable2}
\end{table}
\fi

We present a unified algorithm, \ANACONDA, for both equivalence and uniformity testing, the only difference is in the choice of parameters.
\ANACONDA{} is quite simple to describe, requiring only four sentences in Section~\ref{sec:approach}.\footnote{Perhaps if we tried harder, we could describe it in two sentences, plus the word ``repeat.''}
We consider this algorithmic simplicity to be an advantage of \ANACONDA, though we regret that its analysis is less simple -- we elaborate on the technical challenges in Section~\ref{sec:approach}.

In the following discussion, for ease of exposition, we focus on the case where $\ve$ is a constant.
Our bound for equivalence testing in the \NACOND model is the first tailored to this setting. 
Specifically, the best upper bound was $O(n^{2/3})$ (for the harder problem of equivalence testing in the \SAMP model~\cite{ChanDVV14}), and the best lower bound was $\Omega(\log n)$ (for the easier problem of uniformity testing in the \NACOND model~\cite{AcharyaCK15b}). These results left open the question of the true complexity of equivalence testing: is it polynomial in $\log n$, or polynomial in $n$?
Our algorithm gives an exponential improvement in the query complexity by showing that the former is true: equivalence testing enjoys significant savings in the query complexity when we switch from the \SAMP to the \NACOND oracle model.

More generally, as mentioned before, our results expose a qualitatively interesting relationship between identity and equivalence testing in the \NACOND model.
In the standard sampling model (\SAMP), the complexity of these problems is known to be polynomially related ($\Theta(n^{1/2})$ versus $\Theta(n^{2/3})$).
However, in the conditional sampling model with adaptivity (\COND), there is a ``chasm'' between these two complexities: one has a constant query complexity, while the other has a complexity which is doubly logarithmic in $n$ ($\Theta(1)$ versus $\poly \log \log n$).
Our results demonstrate that when we remove adaptivity from the conditional sampling model (\NACOND), the relationship is qualitatively quite different.
In this setting, the ``chasm'' closes, and the complexity of both problems is once again polynomially related: both are $\poly \log n$. 
Interestingly, this complexity is intermediate to the complexity of the same problems in the \SAMP and \COND models, by an exponential factor on either side.
These relationships are all summarized in Table~\ref{fig:batable}.
We note that our results further address an open problem of Fischer~\cite{Fischer14}, which inquires about the complexity of equivalence testing with conditional samples.

In terms of specific sample complexities, we observe that our upper bound for uniformity testing is nearly tight: our $\tilde O\left(\frac{\log n}{\ve^2}\right)$ upper bound is complemented by the $\Omega(\log n)$ lower bound of~\cite{AcharyaCK15b}. 
It improves upon the algorithm of~\cite{ChakrabortyFGM13}, which has query complexity $O\left(\frac{\log^{12.5}n}{\varepsilon^{17}}\right)$.
Our algorithm for identity testing, with complexity $\tilde O\left(\frac{\log^2 n}{\ve^2}\right)$, also significantly improves over theirs, which has a similar complexity as their algorithm for uniformity testing.
We again mention that our bound for equivalence testing is exponentially better than the previous best algorithm for this problem (which is the $O_\ve(n^{2/3})$-query algorithm in the \SAMP model of~\cite{ChanDVV14}).

\subsection{Techniques and Proof Ideas}
\label{sec:approach}
At the core of our approach is reducing the problem from $\ell_1$-testing to $\ell_\infty$-testing, the latter of which is much cheaper in terms of sample complexity.
In particular, throughout this exposition, keep in mind that one can estimate a distribution up to $\ve$ in $\ell_\infty$-distance at a cost of $O(1/\ve^2)$ samples (cf. Corollary~\ref{cor:linf}).
In order to give intuition on how such an approach could possibly work, we focus on two very simple instances of uniformity testing.
In the first instance, $p$ is a distribution with a single ``spike'': for some $i^* \in [n]$, $p(i^*) = \frac{1}{n} + \ve$, and for $i \neq i^*$, $p(i) = \frac{1-\ve}{n}$.
This can be detected by simply choosing $S = [n]$ and querying it with \NACOND $O(1/\ve)$ times: the empirical distribution $\hat p(i^*)$ will have a similar spike, betraying that the distribution is non-uniform.
In the second instance, $p$ is the ``Paninski construction'' (used as the lower bound in~\cite{Paninski08}): a random half of the domain elements have probability $\frac{1+\ve}{n}$, while the other half have probability $\frac{1-\ve}{n}$.
This can be detected by choosing $S$ to be two random symbols, and again querying this subset $O(1/\ve^2)$ times.
With constant probability, the two symbols will be from different sets.
While the $\ell_\infty$ distance from uniformity on each symbol is only $\frac{\ve}{n}$, in this \emph{conditional} distribution, it is increased to $\ve$, allowing easy detection.

These two examples illustrate the heart of our approach: our algorithm, \ANACONDA, attempts to find a query set in which the discrepancy of a single item is large in comparison to the total probability mass of the set. 
One of our key lemmas (Lemma~\ref{lem:uniform-key}) shows that this is possible with probability $\geq \Omega\left(\frac{1}{\log n} \right)$.
The flavor is somewhat reminiscent of Levin's Economical Work Investment Strategy~\cite{Goldreich14}.
While the two instances above are straightforward, a more careful analysis is required to avoid paying excess factors of $\log n$, particularly for uniformity and identity testing.
That said, all the complexity is pushed to the analysis, and the algorithm itself is very simple to describe:

\begin{quote}
First, the algorithm chooses a uniformly random power of two between $2$ and $n$ -- roughly, this serves as a ``guess'' for (the inverse of) the size of the set which represents the discrepancy between the distributions.
Next, the algorithm chooses a random set $S \subseteq [n]$ of this size.
Finally, it performs $\NACOND$ queries to $S$ (on both distributions, for equivalence testing), in order to form an empirical distribution (which is accurate in $\ell_\infty$-distance) and check whether there is a discrepant symbol or not.
This process is repeated several times, and if we fail to ever detect a discrepant symbol, we can conclude that the distributions are equal.
\end{quote}

Since uniformity testing is relatively well-behaved, the key lemma mentioned above (Lemma~\ref{lem:uniform-key}) does most of the work.
This is because in this setting, once we have a handle on the distribution of the discrepancy, it is easy to reason about how much of the mass from the uniform distribution is contained in a query set.
We require a few additional concentration arguments on the total discrepancy and probability mass contained in the query set, as well as a separate analysis for the case where $|S|$ needs to be small and this concentration does not hold.

We then leverage our algorithm for uniformity testing to provide an algorithm for identity testing.
This uses the reduction of~\cite{ChakrabortyFGM13}\footnote{We note that the \SAMP-model reduction of~\cite{Goldreich16}, from identity testing to uniformity testing, is not known to apply in either the \NACOND or \COND models.}, which partitions the domain so that the conditional distribution on each part is close to uniform, and tests for identity on each part.
This requires a non-adaptive identity tester for distributions which are close to uniform (in $\ell_\infty$-distance) -- we show our analysis for uniformity testing can be adapted to handle this case.
Our application crucially modifies their reduction in order to minimize the sample complexity, as \ANACONDA{} can test against distributions which are further from uniform than theirs.
Specifically, the analysis in Section~\ref{sec:uniform} for uniformity testing actually allows us to test for identity to any distribution which is $O(1/n)$-close to uniform in $\ell_\infty$ distance.
This is in comparison to the weaker tester of~\cite{ChakrabortyFGM13}, which only allows for distributions which are $O(\ve/n)$-close to uniform.

Finally, we turn to the most technically difficult problem of equivalence testing.
This case turns out to be more challenging, as we must simultaneously reason about $p(i), p(S \setminus i), q(i)$, and $q(S \setminus i)$ -- as mentioned prior, it is much easier to control the latter two quantities for uniformity testing.
To establish our result, we must argue that \ANACONDA{} identifies a set $S$ where both differences $p(i) - q(i)$ and $p(S \setminus i) - q(S \setminus i)$ have opposite signs and are simultaneously relatively large compared to the magnitudes of $p(i), p(S \setminus i), q(i)$, and $q(S \setminus i)$ (Proposition~\ref{prop:set}). 
We consider the distribution of the discrepancy $p-q$, with a case analysis depending on the relationship between the ``typical'' magnitudes of the positive and negative discrepancies.
If these magnitudes are close, then we can select a ``smaller'' set $S$ (where ``smaller'' is defined based on these magnitudes) which has a reasonable probability of including a positively and negatively discrepant element of these magnitudes (Lemma~\ref{lem:close}).
On the other hand, if these magnitudes are far, then with an appropriate choice of the size of the set $S$, there is a significant chance that our set will contain an element $i$ with significant positive discrepancy $p(i) - q(i)$, while the total discrepancy in the set $p(S \setminus i) - q(S \setminus i)$ is very negative (Case 2 in Lemma~\ref{lem:main-analysis}).
Despite all these technicalities, we emphasize that the algorithm itself is still quite simple; in particular, it is identical to the algorithm for uniformity testing (modulo some parameter modifications).

\subsection{Organization}
The organization of the paper is as follows.
In Section~\ref{sec:preliminaries}, we cover various preliminary definitions.
In Section~\ref{sec:algorithm}, we unveil \ANACONDA{}.
In Section~\ref{sec:uniform}, we analyze our algorithm for the special case of uniformity testing.
This case is conceptually much simpler than equivalence testing, but exposes some of the key intuitions.
In Section~\ref{sec:equivalence}, we describe the full analysis for equivalence testing.
In Section~\ref{sec:identity}, we adapt the reduction of~\cite{ChakrabortyFGM16} to obtain a more efficient algorithm for identity testing.
We conclude in Section~\ref{sec:open} with some open problems for further investigation.

%!TEX root = main.tex
\section{Preliminaries}
\label{sec:preliminaries}
In this paper, we will focus on discrete distributions over the support $[n]$.
We denote the distributions of interest using $p$ and $q$, where $p(i)$ is the probability placed by distribution $p$ on symbol $i$.
For a set $S$, let $p(S) = \sum_{i \in S} p(i)$.
Furthermore, let $p_S$ be the conditional distribution of $p$ restricted to $S$, i.e., $p_S(i) = p(i)/p(S)$.

We use the following definition of the conditional sampling model.
Note that this uses the convention of~\cite{ChakrabortyFGM13} of sampling uniformly from query sets with $0$ measure, rather than the convention of~\cite{CanonneRS14} which immediately fails if given such a set, as the latter convention trivializes \NACOND, reducing it to \SAMP.
\begin{definition}
\label{def:cond}
A \emph{conditional sampling} oracle for a distribution $p$ is defined as follows: the oracle takes as input a \emph{query set} $S \subseteq [n]$, and returns a symbol $i \in S$, where the probability that $i$ is returned is equal to $p_S(i) = p(i)/p(S)$.
If $p(S) = 0$, then a symbol $i \in S$ is returned uniformly at random.

Given an \emph{adaptive conditional sampling} oracle (a \COND oracle), the algorithm may query \emph{adaptively}: before submitting each query set $i$, the algorithm is allowed to view the results of queries $1$ through $i-1$.
In contrast, given a \emph{non-adaptive conditional sampling} oracle (a \NACOND oracle), the algorithm must be non-adaptive: it must submit all query sets in advance of viewing any of their results.
\end{definition}

We will use the following distances on probability distributions:
\begin{definition}
The \emph{total variation distance} between distributions $p$ and $q$ is defined as $$\dtv(p,q) = \frac12\sum_{i \in [n]} \left|p(i) - q(i)\right|.$$
\end{definition}

\begin{definition}
The \emph{Kolmogorov distance} between distributions $p$ and $q$ is defined as $$\dk(p,q) = \max_{j \in [n]} \left|\sum_{i = 1}^j p(i) - \sum_{i=1}^j q(i) \right|.$$
\end{definition}

%Note that $\ell_\infty$ lower-bounds Kolmogorov distance: in other words, if two probability vectors differ significantly at a single point, then this will result in a large Kolmogorov distance between the two vectors.
%Fortunately, it is quite easy to detect such Kolmogorov discrepancies, as 
The Dvoretzky-Kiefer-Wolfowitz (DKW) inequality gives a generic algorithm for learning any distribution with respect to the Kolmogorov distance~\cite{DvoretzkyKW56}. 
\begin{lemma}[\cite{DvoretzkyKW56},\cite{Massart90}]
 \label{lem:dkw}
  Let $\hat p_m$ be the empirical distribution generated by $m$ i.i.d.\ samples from a distribution $p$.
  We have that 
$$\Pr[\dk(p,\hat p_m) \geq \ve] \leq 2e^{-2m\ve^2}.$$
  In particular, if $m = \Omega(\log(1/\d)/\ve^2)$, then $\Pr[\dk(p, \hat p_m) \geq \ve] \leq \d$.
\end{lemma}

From this, we have the following corollary:
\begin{corollary}
\label{cor:linf}
Let $\hat p_m$ and $\hat q_m$ be the empirical distributions generated by $m = \Omega(\log(1/\d)/\ve^2)$ i.i.d.\ samples from distributions $p$ and $q$.
Then the following occurs with probability at least $1 - \delta$.
If there exists some $i$ such that $|p(i) - q(i)| \geq \ve$, then $|\hat p_m(i) - \hat q_m(i)| \geq 3\ve/4$.
On the other hand, if $p = q$, then for all $i \in [n]$, $|\hat p_m(i) - \hat q_m(i)| \leq \ve/4$.
\end{corollary}
\begin{proof}
Since $\ell_\infty(p,q) \leq 2\dk(p,q)$, Lemma~\ref{lem:dkw} implies that $\ell_\infty(p,\hat p_m) \leq \frac{\ve}{10}$ and $\ell_\infty(q, \hat q_m) \leq \frac{\ve}{10}$.
Both cases follow by triangle inequality.
\end{proof}

We will frequently use $z = (p - q)/\ve$ to denote the ``noise vector'' between $p$ and $q$, and $\bar p = (p + q)/2$.
While the two cases in distribution testing that one considers are usually $p = q$ and $\dtv(p,q) \geq \ve$, for convenience of notation, we will generally assume the latter case to be $\dtv(p,q) = \ve$ -- it is not hard to see that our analysis carries through whenever the algorithm is given a parameter $\ve$ which is less than the true total variation distance between $p$ and $q$.
With this in mind, when $p = q$, we have that $z = \vec 0$, and when $\dtv(p,q) = \ve$, we have that $\|z\|_1 = 2$ and $\sum_{i \in [n]} z(i) = 0$.
Let $z^+$ denote the ``rectified'' version of $z$, where $z^+(i) = \max(0, z(i))$ -- here, in the latter case, $\|z^+\|_1 = \sum_{i \in [n]} z^+(i) = 1$. 
$z^-(i) = \max(0,-z(i))$ is defined similarly.

We will use $\log$ to refer to the logarithm with base $2$ throughout this paper.

For our analysis, we will group indices into bins:

\begin{definition}
\label{def:bin}
The $j$-th bin for a vector $x$, denoted by $\Bin_j(x)$, contains all indices whose values are in the range $[2^{-j}, 2^{-j+1})$, i.e. $\Bin_j(x) \triangleq \{i\ :\ \frac{1}{2^{j}} \leq x(i) < \frac{1}{2^{j-1}}\}$.
\end{definition}

\section{Algorithm}
\label{sec:algorithm}
Our algorithm, \ANACONDA, is presented in Algorithm~\ref{alg:anaconda}. 
While it is phrased in terms of equivalence testing, it still works when a distribution $q$ is explicitly given (i.e., identity testing), as one can simply simulate \NACOND queries to $q$.
It takes three parameters, $T$, $m$, and $\varepsilon'$, which we will instantiate differently (as required by our analysis) for uniformity and equivalence testing.

The algorithm's behavior can roughly be summarized as follows. 
The algorithm first chooses a random size for a query set.
It then chooses a random subset of the domain of this size.
Next, it draws several conditional samples from this set, from both $p$ and $q$.
Finally, if it detects that a single element from the query set has a significantly discrepant probability mass under $p$ and $q$, it outputs that the two distributions are far.
It repeats this process several times, eventually outputting that the distributions are equal if it never discovers a discrepant element.

\begin{algorithm}[htb]
\begin{algorithmic}[1]
\Function{\ANACONDA}{$\ve$, $\NACOND_p$ oracle, $\NACOND_q$ oracle, parameters $T, m, \ve'$}
\For{$t = 1$ to $T$}
\State Choose an integer $j \in \{1,\dots, 2\log n\}$ uniformly at random, and define $r \triangleq 2^j$. \label{ln:select-r}
\State Choose a random set $S \subseteq [n]$, independently selecting each $i$ to be in $S$ with probability $1/r$. \label{ln:select-S}
\State Perform $m$ queries to $\NACOND_p$ and $\NACOND_q$ on the set $S$.
\State Using these queries, form the empirical distribution $\hat p_S$ and $\hat q_S$.
\If {$\exists i \in S$ such that $|\hat p_S(i) - \hat q_S(i)| \geq \ve'$} \label{ln:dkw}
\State \Return $\dtv(p,q) \geq \ve$
\EndIf
\EndFor
\State \Return $p = q$ \label{ln:p-eq-q}
\EndFunction
\end{algorithmic}
\caption{\ANACONDA: An algorithm for testing equivalence given \NACOND oracle access to $p, q$}
\label{alg:anaconda}
\end{algorithm}

\noindent Analyzing the query complexity of this algorithm is straightforward.
\begin{fact}
The query complexity of \ANACONDA is $O(Tm)$.
\end{fact}

\section{Analysis for Uniformity Testing}
\label{sec:uniform}

In this section, we will prove Theorem~\ref{thm:uniform} by instantiating \ANACONDA{} with parameters $T = \Theta(\log n)$, $m = \Theta(\log \log n / \ve^2)$, and $\ve' = \Theta(\ve)$. 

Our strategy will be as follows.
We will argue that, with probability $\Omega(1/\log n)$, \ANACONDA{} will select a set $S$ with a single element that has significantly different mass under the uniform distribution and the distribution $p_S$. 
In this way, we will reduce the problem from $\ell_1$-testing to $\ell_\infty$-testing, the latter of which is solvable with very few samples, by Corollary~\ref{cor:linf}.

More precisely, we compare the probability assigned to a particular symbol $i$ when performing a conditional sample on $S$, in the two cases where $p = \mathcal{U}_n$, and when $\dtv(p,\mathcal{U}_n) = \ve$.
In the former case, the probability is $\frac{\mathcal{U}_n(i)}{\mathcal{U}_n(S)}$, while in the latter, it is $\frac{\mathcal{U}_n(i) + \ve z(i)}{\mathcal{U}_n(S) + \ve z(S)}$.
Therefore, the difference in probability assigned is 
\begin{equation}
\left|\frac{\mathcal{U}_n(i) + \ve z(i)}{\mathcal{U}_n(S) + \ve z(S)} - \frac{\mathcal{U}_n(i)}{\mathcal{U}_n(S)}\right| . \label{eq:discrepancy}
\end{equation}

In the following two subsections, we will show the following lemma:
\begin{lemma}
\label{lem:discrepant-set}
If $\dtv(p,\mathcal{U}_n) = \ve$, then for each $t$, \ANACONDA{} will select a set $S$ which causes (\ref{eq:discrepancy}) to be $\geq \Omega(\ve)$ with probability $\geq \Omega(1/\log n)$.
\end{lemma}
This lemma exposes an interesting structural fact about probability distributions: if a distribution is far from uniform in $\ell_1$-distance, then conditioning on a random set of a random size is likely to give a distribution which is far from uniform in $\ell_\infty$-distance.
One technicality that arises when taking conditional distributions is that the normalization factors might be different.
While not directly applicable to our work, we prove a similar result based on the size of the set in Section~\ref{sec:standalone}, which may be of independent interest.

Assuming this lemma to be true for the moment, we will show how to complete the proof.
Repeating this process $T = \Theta(\log n)$ times will guarantee that at least one iteration will choose an $S$ containing a sufficiently discrepant element with probability $\geq 9/10$.
We focus on the iteration where such an $S$ is selected.

Now if we draw $\Theta(\log \log n/\ve^2)$ samples from $p_S$, Corollary~\ref{cor:linf} implies that the empirical distribution $\hat p_S$ will approximate $p_S$ in $\ell_\infty$ up to an additive $\ve'$, with probability at least $1 - O\left(\frac{1}{\log n}\right)$, and thus Line~\ref{ln:dkw} will correctly identify that $\dtv(p,\mathcal{U}_n) = \ve$.
Therefore, with probability at least $4/5$, the algorithm will correctly detect in this case that $\dtv(p, \mathcal{U}_n) = \ve$.

We now examine what happens when $p = \mathcal{U}_n$.
For each iteration $t$, the uniform distribution on $S$ and $p_S$ will be equal.
We again invoke Corollary~\ref{cor:linf} with $\Theta(\log \log n /\ve^2)$ samples, and use a union bound over all $T = \Theta(\log n)$ iterations.
This implies that, with probability at least $9/10$, Line~\ref{ln:dkw} will never identify an element which has $\geq \ve'$ discrepancy, and thus the algorithm will output that $p = \mathcal{U}_n$ in Line~\ref{ln:p-eq-q}.

It remains to prove Lemma~\ref{lem:discrepant-set}.
We break the analysis into two cases, which we address in the following two subsections.
In Section~\ref{sec:many-small}, we handle the case where, for all $x \in \{z^-,z^+\}$, $\sum_{j = \log (n/32) + 1}^{\log 5n} \sum_{i \in \Bin_j(x)} x(i) \geq 1/5$.
This corresponds to the case where there are many symbols with small discrepancy from the uniform distribution, in both the positive and negative direction.
In Section~\ref{sec:few-small}, we handle the complement of this case, where there exists an $x \in \{z^-,z^+\}$ for which $\sum_{j = 1}^{\log (n/32)} \sum_{i \in \Bin_j(x)} x(i) \geq 3/5$.
Roughly, this happens when there are not too many symbols which capture the discrepancy between the distributions.

Before we proceed, we note the following proposition relating the size of a bin to the mass it contains, which is immediate from Definition~\ref{def:bin}.
\begin{proposition}
\label{prop:bin-size}
  $2^{j-1} \sum_{i \in \Bin_j} x(i) \leq |\Bin_j(x)| \leq 2^j \sum_{i \in \Bin_j} x(i)$.
\end{proposition}

\subsection{Case I: Many Small Discrepancies}
\label{sec:many-small}
In this section, we prove Lemma~\ref{lem:discrepant-set} in the case where for all $x \in \{z^-,z^+\}$, $\sum_{j = \log (n/32) + 1}^{\log 5n} \sum_{i \in \Bin_j(x)} x(i) \geq 1/5$.
In short, the analysis can be summarized as follows: if the algorithm chooses a set $S$ of size $2$, it is likely to contain two elements with non-trivial discrepancy, and in both the positive and negative direction -- this will suffice to make (\ref{eq:discrepancy}) be $\geq \Omega(\ve)$.

Proposition~\ref{prop:bin-size} gives us the following lower bound on the number of symbols which are in bins $\log (n/32) + 1$ through $\log 5n$:
\begin{equation}
\sum_{j = \log (n/32) + 1}^{\log 5n} |\Bin_j(x)| \geq \sum_{j = \log (n/32) + 1}^{\log 5n} 2^{j-1} \sum_{i \in \Bin_j(x)} x(i) \geq \frac{n}{32} \sum_{j = \log (n/32) + 1}^{\log 5n} \sum_{i \in \Bin_j(x)} x(i) \geq \frac{n}{160}
\label{eq:many-small}
\end{equation}
In other words, for either $x \in \{z^-, z^+\}$, there are $\Omega(n)$ symbols with $x(i) \geq 1/5n$.

With probability $\frac{1}{2\log n}$, \ANACONDA{} will select $j = \log n$ in Line~\ref{ln:select-r}.
Conditioning on this, with constant probability, the set $S$ selected in Line~\ref{ln:select-S} will be of size exactly $2$.
Further conditioning on this, due to (\ref{eq:many-small}), with constant probability $S$ will consist of two symbols $i_1 \in \Bin_{j'}(z^+)$ and $i_2 \in \Bin_{j''}(z^-)$ for $\log (n/32) + 1 \leq j', j'' \leq \log 5n$.
To unpack a bit of the notation here: this implies that $z(i_1)$ and $z(i_2)$ will have opposite signs, and are of comparable magnitude.

Without loss of generality, suppose that $z(i_1) \geq 0$ and $z(i_2) \leq 0$.
Then (\ref{eq:discrepancy}) is the following:
\begin{equation}
\left|\frac{\mathcal{U}_n(i_1) + \ve z(i_1)}{\mathcal{U}_n(S) + \ve z(S)} - \frac{\mathcal{U}_n(i_1)}{\mathcal{U}_n(S)}\right| = \frac{\ve n (z(i_i) - z(i_2))}{2(2 + \ve n (z(i_1) + z(i_2)))} \geq \frac{\ve n \cdot \frac{2}{5n}}{2(2 + \ve n \cdot \frac{32}{n})} \geq \frac{\ve}{68}.
\end{equation}
This expression is $\Omega(\ve)$, and this event happens with probability $\geq \Omega(1/\log n)$, thus proving Lemma~\ref{lem:discrepant-set} in this case.

\subsection{Case II: Not So Many Small Discrepancies}
\label{sec:few-small}
In this section, we prove Lemma~\ref{lem:discrepant-set} in the case where there exists an $x \in \{z^-,z^+\}$ for which $\sum_{j = 1}^{\log (n/32)} \sum_{i \in \Bin_j(x)} x(i) \geq 3/5$.
Without loss of generality, assume that this holds for $z^+$.
Furthermore, we focus our analysis on the case where \ANACONDA{} picks an $j \leq \log (n/32)$.
For the remainder of this proof we condition on this event, which happens with probability at least $1/4$.

We will need the following key lemma:
\begin{lemma}
\label{lem:uniform-key}
Suppose $\dtv(p,\mathcal{U}_n) = \varepsilon$.
For each iteration $t$, with probability $\geq \frac{3}{20 \log (n/32)}$, the algorithm will choose an $r$ and a set $S$ such that there exists $i \in S$ with $z^+(i) \geq 1/r$.  
\end{lemma}
\begin{proof}
For some fixed $j$, the probability of choosing $j$ is $\frac{1}{\log (n/32)}$, and, conditioning on this $j$, the probability of picking any element from $\Bin_j(z^+)$ to be in $S$ is $1 - \left(1 - \frac{1}{2^j}\right)^{|\Bin_j(z^+)|}$.
By the law of total probability, we sum this over all bins to get the probability that the event of interest happens:
\begin{align} \frac{1}{\log (n/32)} \sum_{j \in [\log (n/32)]} 1 - \left(1 - \frac{1}{2^j}\right)^{|\Bin_j(z^+)|} 
&\geq \frac{1}{\log (n/32)} \sum_{j \in [\log (n/32)]} 1 - \exp\left(-\frac{|\Bin_j(z^+)|}{2^j}\right) \label{eq:noiselem1} \\
&\geq \frac{1}{\log (n/32)} \sum_{j \in [\log (n/32)]} 1 - \exp\left(-\frac12\sum_{i \in \Bin_j(z^+)} z^+(i) \right) \label{eq:noiselem2} \\
&\geq \frac{1}{\log (n/32)} \sum_{j \in [\log (n/32)]} \frac14\sum_{i \in \Bin_j(z^+)} z^+(i) \label{eq:noiselem3} \\
&\geq \frac{3}{20 \log (n/32)}. \label{eq:noiselem4}
\end{align}
(\ref{eq:noiselem1}) follows from the inequality $1 - x \leq \exp(-x)$, (\ref{eq:noiselem2}) is due to Proposition~\ref{prop:bin-size}, (\ref{eq:noiselem3}) is by the inequality $1 - \exp(-x) \geq x/2$ (which holds for all $x \in [0,1]$), and (\ref{eq:noiselem4}) is by assumption.
\end{proof}

We will require the following lemmata to complete the proof:
\begin{lemma}
\label{lem:rest-signal}
For all $n$ greater than some absolute constant, for any $i \in [n]$ and $j \leq \log (n/32)$, $$\Pr\left[\frac{1}{2\cdot 2^j} \leq \mathcal{U}_n\left(S \setminus \{i\}\right) \leq \frac{3}{2 \cdot 2^j}\right] \geq 1 - 2/e^2.$$
\end{lemma}
\begin{proof}
Observe that the size of $S\setminus \{i\}$ is a sum of $n-1$ i.i.d.\ Bernoulli random variables with parameter $1/2^j$, and thus has expectation $\mu = \frac{n-1}{2^j} \geq 32\frac{n-1}{n}$, where the inequality follows by the upper bound on $j$.
Then, by a Chernoff bound, we have 
$$\Pr\left[\frac{9}{16}\frac{(n-1)}{2^j} \leq |S \setminus \{i\}| \leq \frac{23}{16} \frac{(n-1)}{2^j}\right] \geq 1 - 2\exp\left(-\frac{49\mu}{768}\right) \geq 1 - 2\exp\left(-2.04\frac{n-1}{n}\right)\geq 1 - 2/e^2.$$
The final inequality follows for all $n$ larger than some absolute constant.
Note that this statement deals with the size of the set $S \setminus \{i\}$, whereas the lemma statement is concerned with the measure of the set under the uniform distribution: this is simply $\frac{|S \setminus \{i\}|}{n}$, so 
$$\Pr\left[\frac{9}{16}\frac{n-1}{n} \frac{1}{2^j} \leq \mathcal{U}_n\left(S \setminus \{i\}\right) \leq \frac{23}{16} \frac{n-1}{n} \frac{1}{ 2^j}\right] \geq 1 - 2/e^2.$$
Note that for $n$ larger than some absolute constant, $\frac{9}{16}\frac{n-1}{n} \geq 1/2$, as desired.
The bound on the other side follows similarly.
\end{proof}

\begin{lemma}
\label{lem:rest-noise}
If $\dtv(p, \mathcal{U}_n) = \ve$, then for any $i$ and $j$, $$\Pr\left[z\left(S \setminus \{i\}\right) \geq \frac{4}{2^j}\right] \leq 1/4.$$
\end{lemma}
\begin{proof}
Note that $z^+(S \setminus \{i\})$ is a non-negative random variable.
Its expectation $\E[z^+(S \setminus \{i\})] \leq \E[z^+(S)] \leq 1/2^j$.
The lemma follows by Markov's inequality, and by observing that the addition of any negative elements of $z$ will only decrease $z(S \setminus \{i\})$. 
\end{proof}

Note that, by Lemmas~\ref{lem:uniform-key}, \ref{lem:rest-signal}, \ref{lem:rest-noise}, if $\dtv(p, \mathcal{U}_n) = \ve$, with probability at least $\frac14 \cdot \frac{1}{\log (n/32)}\cdot \left(1 - \frac14 - 2/e^2\right) \geq \Omega(1/\log n)$, the following events happen simultaneously for some $i \in S$:
\begin{itemize}
\item $r \leq n/32$;
\item $z(i) \geq 1/r$;
\item $\mathcal{U}_n(i) = 1/n$;
\item $z(S \setminus \{i\}) \leq 4/r$;
\item $\frac{1}{2r} \leq \mathcal{U}_n(S \setminus \{i\}) \leq \frac{3}{2r}$;
\end{itemize}

We now show that a set $S$ with all these properties will result in (\ref{eq:discrepancy}) being $\geq \Omega(\ve)$:
\begin{align*}
\left|\frac{\mathcal{U}_n(i) + \ve z(i)}{\mathcal{U}_n(S) + \ve z(S)} - \frac{\mathcal{U}_n(i)}{\mathcal{U}_n(S)}\right| &= \ve \left|\frac{z(i)\mathcal{U}_n(S \setminus \{i\}) - z(S\setminus \{i\})\mathcal{U}_n(i)}{\mathcal{U}_n(S)(\mathcal{U}_n(S) + \ve z(S))}\right| \\
&\geq \ve \cdot  \frac{1}{\mathcal{U}_n(S) (\mathcal{U}_n(S) + \ve z(S))} \left( \frac{z(i)}{2r} - \frac{4}{rn} \right) \\
&\geq \ve \cdot  \frac{r}{2} \frac{1}{\frac{2}{r} + \ve \left(\frac4{r} + z(i)\right)} \left( \frac{z(i)}{2r} - \frac{4}{rn} \right) \\
&\geq \ve \cdot  \frac{1}{\frac{2}{r} + \ve \left(\frac4{r} + z(i)\right)} \left( \frac{z(i)}{4} - \frac{2}{n} \right) 
\end{align*}
The analysis concludes by considering two cases.
If $\ve z(i) \geq \frac{2}{r} + \ve \cdot \frac{4}{r}$, then we have the lower bound $\ve \cdot \frac{1}{2\ve z(i)} \left(\frac{z(i)}{4} - \frac{2}{n} \right) = \Omega(1) \geq \Omega(\ve)$, as desired.
Otherwise, we have the lower bound $\ve \cdot  \frac{r}{12} \left( \frac{z(i)}{4} - \frac{2}{n} \right) \geq \ve \cdot \frac{r}{12}\left(\frac{1}{4r} - \frac{2}{n}\right) \geq \frac{\ve}{96}$, which completes the proof.

%!TEX root = main.tex
\section{Analysis for Equivalence Testing}
\label{sec:equivalence}
In this section, we will prove Theorem~\ref{thm:equivalence} by instantiating \ANACONDA{} with parameters $T = \Theta(\log^6 n)$, $m = \tilde \Theta(\log^6 n / \ve^2)$, and $\ve' = \frac{\ve}{\tilde \Theta(\log^3 n)}$. 

%To prove the theorem, we will show that if $\dtv(p,q) \geq \ve$, Algorithm~\ref{alg:anaconda} will detect it with high probability. To show this, we will use the following proposition: 

We will require the following proposition, which says that if $\dtv(p,q) = \ve$ and \ANACONDA{} selects an appropriate set $S$, then it will detect the discrepancy.
\begin{proposition}\label{prop:set}
  Suppose that $\dtv(p,q) = \ve$ and that within the first $T$ iterations a set $S$ is identified such that for some $i \in S$ and some $c > 0$,
  $$ \min\{z(i), z(i) - z(S)\}  \ge \frac {p(S) + q(S)}  {\tilde O(\log^c n)}. $$
  Then, for $\eps' = \frac {\eps} {\tilde O(\log^c n)}$ and $m = \tilde \Omega\left(\frac{ \log^{2 c} n } { \eps^2 } \right)$, the algorithm outputs that $\dtv(p,q) \geq \ve$ with probability at least $1 - \frac 1 {\poly \log n}$.
\end{proposition}

\begin{proof}

We first argue that $|p_S(i) - q_S(i)| \ge \eps \frac { \min\{z(i), z(i) - z(S)\} } {p(S) + q(S)}$ for some $i \in S$.

We set $\bar p = \frac {p+q} 2$. We have that $p = \bar p + z \frac \eps  2$, $q = \bar p - z \frac \eps  2$ and
$$\left| \frac {p(i)} {p(S)} -  \frac {q(i)} {q(S)}\right| = \left| \frac {\bar p(i) + z(i) \frac \eps  2} {\bar p(S) + z(S) \frac \eps  2} -  \frac {\bar p(i) - z(i) \frac \eps  2} {\bar p(S) - z(S) \frac \eps  2}\right| = 
\frac \eps 2 \left| \frac {z(i) \bar p(S) - \bar p(i) z(S)} {\bar p^2(S) - ( z(S) \frac \eps  2)^2} \right| \ge 
\frac \eps 2 \left| \frac {z(i) \bar p(S) - \bar p(i) z(S)} {\bar p^2(S)} \right|. $$

As $z(i) \bar p(S) - \bar p(i) z(S) \ge \bar p(S) \min\{z(i), z(i) - z(S)\}$, it follows that $$|p_S(i) - q_S(i)| \ge \frac \eps 2 \frac { \min\{z(i), z(i) - z(S)\} } {(p(S) + q(S))/2} $$

To complete the proof, we note that the condition implies that 
$|p_S(i) - q_S(i)| \ge \frac {\eps}  {\tilde O(\log^c n)}$ and thus by Corollary~\ref{cor:linf}, $m = \tilde O\left(\frac{ \log^{2 c} n } { \eps^2 } \right)$ suffices to detect (with failure probability $1 / \poly\log n$) that $\|p_S - q_S\|_{\infty} > \eps' = \frac {\eps}  {\tilde O(\log^c n)}$.
\end{proof}

To complete the proof, we will show that after $T=\poly\log n$ iterations, Algorithm~\ref{alg:anaconda} will choose a set $S$ that satisfies the conditions of Proposition~\ref{prop:set}.

We define $\hat z$ to be the vector with $\hat z(i) = z(i)$ if $|z(i)| > \frac {p(i) + q(i) }{ 400 \log n } $ and $\hat z(i) = 0$ otherwise. 
Roughly, this ``zeroes out'' the noise for any $i$ where the noise vector $z$ is too small in comparison to the signal vector $p + q$.
Let $b^+$ be the measure on $\{1,\dots, 2 \log n\}$ with mass $\hat z^+( \Bin_j(z^+) )$ and equivalently define $b^{-}$. Notice that $|b^+|, |b^-| \in [1-\frac {1} {200 \log n}, 1]$. This is because $\sum_{i: \hat z^+(i) = 0} z^+(i) \le \sum_i \frac {p(i) + q(i) }{ 400 \log n } \le \frac {1} {200 \log n}$.

The next lemma shows that, if there are two bins (with respect to the positive and negative $z$ vectors) which are both ``heavy'' and are close in index, then we will obtain an appropriate set $S$ (for Proposition~\ref{prop:set}).
\begin{lemma}\label{lem:close}
  Let $\alpha, \beta, \gamma > 0$ be arbitrary constants. 
  If $b^+(j) > \frac 1 {\tilde O(\log^\alpha n)}$ and $b^-(j') > \frac 1 {\tilde O(\log^\beta n)}$, for some $j$ and $j'$ with $2^{|j-j'|} = \tilde O(\log^\gamma n)$, then a single iteration of Algorithm~\ref{alg:anaconda} finds set $S$ and $i \in S$ with 
  $ \min\{z(i), z(i) - z(S)\}  \ge \frac {p(S) + q(S)}  { \tilde O( \log^{\gamma + 1} n ) } $ with probability $\frac 1 {\tilde O(\log^{\alpha + \beta + \gamma + 1} n)}$.
\end{lemma}

\begin{proof}
  With probability $\frac 1 {\tilde O(\log n)}$, an iteration of Algorithm~\ref{alg:anaconda} will choose $r = 2^{-\max\{j,j'\}-3}$. Given this value of $r$, a unique $i$ with $\hat z^+(i) \in [2^{-j}, 2^{-j+1})$  and a unique $i'$ with $\hat z^-(i') \in [2^{-j'}, 2^{-j'+1})$ are selected with probability $\frac 1 {\tilde O(\log^{\alpha + \beta + \gamma} n)}$. It holds that $z^-(i'), z^+(i) \in [8,  \tilde O(\log^\gamma n)] \cdot r$ and their corresponding $p(i)+q(i) \le O( \log n ) \cdot z(i) \le \tilde O(\log^{1+\gamma} n) r$ and $p(i')+q(i') \le \tilde O(\log^{1+\gamma} n) r$.
  
By Markov's inequality, with probability at least $3/4$, $z(S \setminus \{i,i'\}) \le z^{+}(S \setminus \{i,i'\}) \le 4 r$. Similarly, with probability at least $3/4$, $p(S \setminus \{i,i'\}) + q(S \setminus \{i,i'\}) \le 8 r$. By a union bound with probability $1/2$ both hold simultaneously.

When all of these events occur, which happens with probability at least $\frac 1 {\tilde O(\log^{\alpha + \beta + \gamma + 1} n)}$ we get that:

$$\min\{z(i), z(i) - z(S)\} \ge 4 r \quad \text{ since } \quad z(i) - z(S) \ge z^{-}(i') - z(S \setminus \{i,i'\}) \ge 4 r$$

The lemma follows by noting that  $p(S) + q(S) \le \tilde O(\log^{\gamma + 1} n) r$.
\end{proof}

Finally, we have our main lemma required for the analysis.
It leverages Lemma~\ref{lem:close} to show that we can obtain an appropriate set $S$ with reasonable probability.

\begin{lemma}\label{lem:main-analysis}
  If $\dtv(p,q) = \ve$, then a single iteration of Algorithm~\ref{alg:anaconda} finds set $S$ and $i \in S$ with 
  $ \min\{z(i), z(i) - z(S)\}  \ge \frac {p(S) + q(S)}  { \tilde O( \log^{3} n ) } $ with probability $\frac 1 {\tilde O(\log^{6} n)}$.
\end{lemma}

Before we prove Lemma~\ref{lem:main-analysis}, we require the following two simple concentration lemmas:
\begin{lemma}\label{lem:concentration}
  Let $0 < a < b$, $X_i \sim Bernoulli(2^{-a})$ and let $1 > \sum_{i:x_i<2^{-b}} x_i \ge c$. Then, $\sum_{i:x_i<2^{-b}} X_i x_i > 2^{-a} ( c -  t 2^{-(b-a)/2} )$, with probability $1 - e^{-t}$.
\end{lemma}

\begin{proof}
  We apply the Chernoff bound on the variables $Z_i = X_i 2^b x_i$. We get that with probability $1-e^{-t}$, $2^{b} \sum_{i:x_i<2^{-b}} X_i x_i > 2^{b-a} c - t 2^{(b-a)/2}$.
  Thus, $2^{a} \sum_{i:x_i<2^{-b}} X_i x_i > c - t 2^{-(b-a)/2}$.
\end{proof}

\begin{lemma}\label{lem:highprob}
  Let $a \ge 1$, $X_i \sim Bernoulli(2^{-a})$ and let $1 > \sum_{i:x_i>2^{-a}} x_i$. Then, $\sum_{i:x_i>2^{-a}} X_i x_i = 0$, with probability $\frac 1 4$.
\end{lemma}

\begin{proof}
  There are at most $2^a$ elements $x_i$ and every element is selected independently with probability $2^{-a}$. The probability that no element is chosen is
  $(1- 2^{-a})^{2^a} \ge \frac 1 4$.
\end{proof}

We continue with the proof of Lemma~\ref{lem:main-analysis}. 

\begin{prevproof}{Lemma}{lem:main-analysis}
We consider two cases:

\begin{enumerate}
  \item $\dk(b^+,b^-) \le \frac 1 {8 \log n}$.
  
  In this case, as $\sum b^+(j) > 2/3$, there will be a bin $j$ with $b^+(j) \ge \frac {2/3} {2 \log n}$.
  As the $\dk(b^+,b^-) \le \frac 1 {8 \log n}$, the corresponding $b^-(j) \ge \frac 1 {3 \log n} - \frac 2 {8 \log n} \ge \frac 1 {12 \log n}$. Then, Lemma~\ref{lem:close} implies that a good set will be identified with high probability.

  \item $\dk(b^+,b^-) > \frac 1 {8 \log n}$.
  
  In this case, there will be a bin $j_r$ with $|\sum_{j \ge j_r} b^-(j) - \sum_{j \ge j_r} b^+(j)| \ge \frac 1 {8 \log n}$. Without loss of generality, $\sum_{j \ge j_r} b^+(j) < \sum_{j \ge j_r} b^-(j)$. 
  
  Let $j_l$ be the largest index such that $\frac 1 {16 \log n} < \sum_{j = j_l}^{j_r} b^+(j)$.
  Then there must exist a $j^* \in [j_l, j_r]$ such that $b^+(j^*) > \frac 1 {32 \log^2 n}$ as $|[j_l, j_r]| \le 2 \log n$.
  
  If there is a $j \in [j^*, j^* + 2 \log \log n]$, with $b^-(j) > \frac 1 {100 \log n \log \log n}$, Lemma~\ref{lem:close} implies that 
  with probability $\frac 1 {O(\log^6 n)}$, $\min\{z(i), z(i) - z(S)\}  \ge \frac {p(S) + q(S)}  { \tilde O( \log^{3} n ) } $.
  %a good set will be identified with high probability.
  
  Otherwise, we have that 
  $$\sum_{j \ge j^* + 2 \log \log n} b^-(j)  > \sum_{j \ge j_r} b^-(j) -  \frac 1 {100 \log n} > \sum_{j > j_l} b^+(j) + \frac 1 {16 \log n} - \frac 1 {100 \log n} > \sum_{j > j^*} b^+(j) + \frac 1 {20 \log n}. $$ 
  
  We will show that in this case, when the algorithm selects $r = 2^{-j^*}$, a good set $S$ is identified with non-trivial probability. We bound the contribution at the contribution from every bin.
  
  \begin{enumerate}
    \item From bin $b^+(j^*)$: With probability $\Omega( b^+(j^*) ) = \frac 1 {O(\log^2 n)}$, a unique $i$ with $\hat z^+(i) \in [2^{-j^*}, 2^{-j^*+1})$ is selected in $S$. It holds that $z^+(i) \in [r,  2r]$.
    
    \item From bin $b^+(j)$ for $j>j^*$:  $\hat z^+( S \cap (\bigcup_{j > j^*} \Bin_j(z^+)) \} ) \le A r \sum_{j > j^*} b^+(j) $ with probability at least $1-1/A$ for any $A \ge 1$. This holds by Markov's inequality. Setting $A = 1 + \frac 1 {200 \log n}$, we get that with probability at least $\frac 1 {200 \log n + 1}$, 
    $$\hat z^+( S \cap (\bigcup_{j > j^*} \Bin_j(z^+)) \} ) \le r \sum_{j > j^*} b^+(j) + \frac r {200 \log n} . $$
    
    \item From bin $b^+(j)$ for $j<j^*$: $\hat z^+( S \cap (\bigcup_{j < j^*} \Bin_j(z^+)) \} ) = 0$ with probability $1/4$. This holds by Lemma~\ref{lem:highprob}.

    \item From bin $b^-(j)$ for all $j$: $z^-(S) \ge r \sum_{j \ge j^* + 2 \log \log n } b^-(j) - \frac r {200 \log n}$ with probability $15/16$. This holds by a concentration bound presented in Lemma~\ref{lem:concentration}. We also have that $\hat z^-(S) \le 2 r$ with probability $1/2$. By a union bound both inequalities hold with probability at least $7/16$.
    
\end{enumerate}

Finally we consider the contribution from all elements $i$ not accounted for. These are all the elements in the set $Z =\left\{i:  |z(i)| \le \frac {p(i) + q(i) }{ 400 \log n } \right\}$.
  
\begin{enumerate}
\item[(e1)] We have that $z^+( S \cap Z ) \le 4 r z^+( Z ) $ with probability $1/4$. This holds by Markov's inequality.
\item[(e2)] We similarly have $p(S \cap Z ) + q(S \cap Z )  \le 4 r \big( p(Z ) +  q(Z ) \big) \le 8 r$ with probability $1/4$. 
\end{enumerate}
By a union bound, the event (e) that includes both events (e1) and (e2) holds w.p. at least $1/2$.
  
Events (a),(b),(c),(d), and (e) are all independent, so all conditions hold with probability $\frac 1 {O(\log^3 n)}$.
  
When these occur we get that

\begin{align*}
  - z(S \setminus \{i\}) &\ge r \sum_{j \ge j^* + 2 \log \log n } b^-(j) - r \sum_{j > j^*} b^+(j) - \frac {r} {100 \log n} - 4 r z^+( Z ) \\
  &\ge \frac r {20 \log n} - \frac {r} {100 \log n} - \frac {r} {50 \log n} \ge \frac r {50 \log n}.
\end{align*} 
  
  In addition, $z(i) \ge r$ and thus $ \min\{z(i), z(i) - z(S)\} \ge \frac r {50 \log n}$. We also have that $p(S) + q(S) \le 
  p(S \cap Z ) + q(S \cap Z ) + p(S \cap \bar Z ) + q(S \cap \bar Z ) \le
  O(\log n) \cdot r$. The last equality holds by the bound (e2) on $p(S \cap Z ) + q(S \cap Z )$ and the fact that $p(S \cap \bar Z ) + q(S \cap \bar Z ) < O(\log n) |z(S \cap \bar Z)| \le O(\log n) \cdot r$
  
  Thus, with probability $\frac 1 {O(\log^3 n)}$, $\min\{z(i), z(i) - z(S)\}  \ge \frac {p(S) + q(S)}  { \tilde O( \log^{2} n ) } $.
  
\end{enumerate}

\end{prevproof}

Finally, with Lemma~\ref{lem:main-analysis} in hand, we combine it with Proposition~\ref{prop:set} to complete the proof of Theorem~\ref{thm:equivalence}.

\begin{prevproof}{Theorem}{thm:equivalence}
  Set $T = \Theta(\log^6(n))$. Then Lemma~\ref{lem:main-analysis} implies that, with constant probability, after $T$ iterations, a set $S$ will be identified such that for some $i \in S$,
  $$ \min\{z(i), z(i) - z(S)\}  \ge \frac {p(S) + q(S)}  {\tilde O(\log^3 n)}. $$
  
Proposition~\ref{prop:set} then implies that for $\eps' = \frac {\eps} {\tilde O(\log^3 n)}$ and $m = \tilde \Omega\left(\frac{ \log^{6} n } { \eps^2 } \right)$, the algorithm outputs that $\dtv(p,q) \geq \ve$ with probability at least $1 - \frac 1 {\poly \log n}$.

In contrast, when $\dtv(p,q) = 0$, the algorithm incorrectly outputs that $\dtv(p,q) \geq \ve$ with probability at most $\frac 1 {\poly \log n}$.
\end{prevproof}

\section{Analysis for Identity Testing}
\label{sec:identity}
In this section, we discuss how our results for uniformity testing imply Theorem~\ref{thm:identity} for identity testing.
We adapt the reduction of~\cite{ChakrabortyFGM16}, from non-adaptive identity testing to non-adaptive near-uniform identity testing.
As before, identity testing is the problem of testing whether an unknown distribution $p$ is equal to a known distribution $q$, or $\ve$-far from it in total variation distance.
%We define $\alpha$-near-uniform identity testing as follows:
%\begin{itemize}
%\item \bf{$\alpha$-near-uniform Identity Testing:} Given sample access to a distribution $p$ and the description of a distribution $q$ such that $\ell_\infty(q,\mathcal{U}_n) \leq \alpha$, decide whether $p = q$ or $\dtv(p,q) \geq \ve$.
%\end{itemize}
In near-uniform identity testing, we are only concerned with testing identity to those $q$ who are ``sufficiently close'' to the uniform distribution in $\ell_\infty$-distance.
We use Algorithm 4.2.2 of~\cite{ChakrabortyFGM16}, with a few crucial differences.
First, for those familiar with their paper, in the following paragraph we summarize the diffs required to obtain our algorithm from that of~\cite{ChakrabortyFGM16}.
Following that, for completeness, we state the algorithm with these diffs implemented.

First, in Line 1, they partition the domain using $Bucket(q, [n], \frac{\ve}{30})$.
We perform a less fine-grained partitioning, using $Bucket(q, [n], \frac{1}{100})$.
Additionally, their bucketing defines $M_0$ as all $i$ such that $q(i) < \frac{1}{n}$.
We define it as all $i$ such that $q(i) < \frac{\ve}{100n}$.\footnote{We note that the original definition of $M_0$ used in~\cite{ChakrabortyFGM13,ChakrabortyFGM16} appears to be an erratum, and a similar modification is required for the reduction to go through in their setting as well.}
The first modification will require a stronger near-uniform identity tester than the one in their paper, which can handle identity testing to any distribution $q$ such that $\|q - \mathcal{U}_n\|_\infty \leq \frac{1}{100n}$.
The second change implies that we do not have to do a near-uniform identity test on $M_0$ -- either $\|z(M_0)\|_1 > 1/50$ and the discrepancy will be discovered in Line 3, or $\|z(M_0)\|_1 \leq 1/50$, and this bucket can be ignored, as $\|z([n] \setminus M_0)\|_1 \geq 49/50$.
As a result of these changes, there are only $\Theta(\log (n/\ve))$ buckets in the partition, and we perform the tests in Line 2 with error bound $\frac{\d \log (1 + 1/100) }{2\log (100n/\ve)}$.

The algorithm for \NACOND identity testing is given in Algorithm~\ref{alg:newid}.
It relies upon a subroutine $Bucket$, which is described in Definition~\ref{def:bucket}.

\begin{definition}[Modification of Definition 2.2.5 of~\cite{ChakrabortyFGM16}]
\label{def:bucket}
Given an explicit distribution $q$ over $[n]$, $Bucket(q, [n], \tau, \ve)$ is a procedure that generates a partition $\{M_0, \dots, M_k\}$ of the domain $[n]$, where $k = \frac{\log (n/\tau)}{\log (1+\ve)} \leq \frac{2}{\ve} \log (n/\tau)$. This partition satisfies the following conditions:
\begin{itemize}
\item $M_0 = \{i \in [n]\ |\ q(i) < \frac{\tau}{n}\};$
\item for all $j \in [k]$, $M_j = \{i \in [n]\ |\ \frac{\tau(1+\ve)^{j-1}}{n} \leq q(i) \leq \frac{\tau(1+\ve)^j}{n}\}$.
\end{itemize}
\end{definition}

\begin{algorithm}[htb]
\begin{algorithmic}[1]
\Function{\textsc{NonAdaptiveIdentity}}{$\ve$, $\delta$, $\NACOND_p$ oracle, description of $q$}
\State Let $\mathcal{M} = \{M_0, M_1, \dots, M_k\}$ be the output of $Bucket(q, [n], \frac{\ve}{100}, \frac{1}{100})$.
\State For each bucket $M_1, \dots, M_k$, use an \NACOND $\frac{1}{100n}$-near-uniform identity test to determine whether there exists a $j$ such that $\dtv(p_{M_j}, q_{M_j}) \geq \frac{\ve}{2}$ with error probability $\frac{\delta \log(1 + 1/100)}{2\log(100n/\ve)}$. \label{ln:near-uniform-test}
\State If any such $j$ is found, \Return $\dtv(p,q) \geq \ve$.
\State Let $\tilde p$ and $\tilde q$ be distributions over $\{0, 1, \dots, k\}$, where $\tilde p(i) = p(M_j)$ and $\tilde q(i) = q(M_j)$.
\State Use a \SAMP identity test to determine whether $\tilde p = \tilde q$ or $\dtv(\tilde p, \tilde q) \geq \ve/2$, and \Return the corresponding answer.
\EndFunction
\end{algorithmic}
\caption{An algorithm for testing identity to $q$ given \NACOND oracle access to $p$}
\label{alg:newid}
\end{algorithm}

With these changes, mimicking the analysis of Theorem 4.2.1 of~\cite{ChakrabortyFGM16} gives the following theorem:
\begin{theorem}
\label{thm:id-to-uniform}
Suppose there exists an $m(n,\ve,\delta)$-query algorithm, which, given \NACOND access to an unknown distribution $p$ over $[n]$ and a description of a distribution $q$ over $[n]$ such that $\|q - \mathcal{U}_n\|_\infty \leq \frac{1}{100n}$, distinguishes between the cases $p = q$ versus $\dtv(p,q) \geq \ve$ with probability $1 - \delta$.

Then there exists an algorithm which, given \NACOND access to an unknown distribution $p$ on $[n]$ and a description of a distribution $q$, makes $O\left(\log (n/\ve) \cdot m\left(n, \ve/2, \frac{\log (1 + 1/100)}{6\log (100n/\varepsilon)}\right) + \frac{\sqrt{\log (n/\ve)}}{\ve^2}\right)$ queries to the oracle on $p$ and distinguishes between the cases $p = q$ and $\dtv(p,q) \geq \ve$ with probability at least $2/3$.
\end{theorem}
While we omit the correctness of the algorithm (the interested reader can refer to~\cite{ChakrabortyFGM16}), the query complexity is straightforward to analyze.
There are $k = \Theta(\log(n/\ve))$ calls to the \NACOND near-uniform identity tester (with the corresponding parameters), followed by \SAMP identity testing on a support of size $k+1 = \Theta(\log (n/\ve))$, for which the optimal sample complexity is $\Theta\left(\frac{\sqrt{\log (n/\ve)}}{\ve^2}\right)$~\cite{Paninski08, ValiantV17a}.

It remains to show the existence of such an \NACOND near-uniform identity tester, as required by Line~\ref{ln:near-uniform-test} of Algorithm~\ref{alg:newid}.
In the rest of this section, we will sketch how the analysis of Theorem~\ref{thm:uniform} can be extended to apply to any distribution $q$ such that $\|q - \mathcal{U}_n\|_\infty \leq \frac{1}{100n}$, while maintaining the same sample complexity:
\begin{theorem}[Non-Adaptive Near-Uniform Identity Testing]
\label{thm:near-uniform}
There exists an algorithm which, given \NACOND access to an unknown distribution $p$ over $[n]$ and a description of a distribution $q$ over $[n]$ such that $\|q - \mathcal{U}_n\|_\infty \leq \frac{1}{100n}$, makes $\tilde O\left(\frac{\log n}{\ve^2}\right)$ queries to the oracle on $p$ and distinguishes between the cases $p = q$ versus $\dtv(p,q) \geq \ve$ with probability at least $2/3$.
\end{theorem}
With this in hand, instantiating Theorem~\ref{thm:id-to-uniform} with $m(n, \ve, \d) = \tilde O\left(\frac{\log n}{\ve^2} \cdot \log (1/\delta)\right)$\footnote{Note that a standard boosting applied to Theorem~\ref{thm:near-uniform} gives a $1 - \delta$ probability of success at a multiplicative cost of $\log (1/\delta)$.} gives Theorem~\ref{thm:identity}.

Most of the analysis in Section~\ref{sec:uniform} involves reasoning about the noise vector $z$, none of which changes for this setting.
The exceptions are at the end of Sections~\ref{sec:many-small} and~\ref{sec:few-small}, where we argue that (\ref{eq:discrepancy}) is large.
We deal with the former case first -- here, (\ref{eq:discrepancy}) can be written as
$$\ve \cdot \left|\frac{z(i_1)q(i_2) - z(i_2)q(i_1)}{q(S)(q(S) + \ve z(S))}\right| \geq \ve \frac{2 \cdot \frac{1}{5n} \cdot \frac{99}{100n} }{\frac{202}{100n}\left(\frac{202}{100n} + \ve \cdot  \frac{32}{n}\right)} \geq \Omega(\ve),$$
as desired.
In the latter case, the proof follows with two minimal changes in the events that happen simultaneously (mentioned towards the end of the section).
Instead of $\mathcal{U}_n(i) = 1/n$, we have that $q(i) \leq 101/100n$.
Also, instead of $\frac{1}{2r} \leq \mathcal{U}_n(S \setminus i) \leq \frac{3}{2r}$, we have that $\frac{1}{2r} \leq q(S \setminus i) \leq \frac{3}{2r}$.
This can be proved by essentially the same argument as Lemma~\ref{lem:rest-signal}, but rescaling at the end by a factor of $100n/99$ or $100n/101$.
With these changes, the argument is identical, and thus we have Theorem~\ref{thm:near-uniform}, implying Theorem~\ref{thm:identity}.

\section{Open Problems}
\label{sec:open}
In this paper, we managed to attain improved upper bounds for several testing problems in the \NACOND model.
However, there is still much room for improvement, since our upper bounds only match the lower bounds for the case of uniformity testing, where the complexity is known to be $\tilde \Theta(\log n)$.

A first question is to sharply characterize the complexity of general identity testing.
While in the \SAMP model, uniformity testing is known to be complete for identity testing~\cite{Goldreich16}, a moment's thought indicates that the same reduction does not immediately hold for either the \COND or \NACOND model.
This is (roughly) because Goldreich's reduction involves mapping the problem onto a larger domain, which would require more ``granular'' conditional samples than afforded by standard conditional sampling models in order for the reduction to go through. 
Therefore, it is plausible that testing identity to a general distribution $q$ is \emph{harder} than uniformity testing -- this would be a qualitative difference in complexity which we are not aware of in any other sampling model.

Naturally, another question is to characterize the query complexity of equivalence testing.
There are several possibilities here -- it may be the same as that of uniformity or identity testing, or distinct from both. 
We would consider either of the former two to be surprising, as this would be qualitatively different behavior than either of the two neighboring oracle models (\SAMP and \COND). 

Beyond the problems considered in this paper, there are many other unexplored questions in the world of distribution testing with conditional samples.
One problem is testing independence of random variables, which has not been considered at all in the conditional sampling model.
In \SAMP, the complexity of this problem is necessarily exponential in the dimension~\cite{AcharyaDK15,DiakonikolasK16}.
An interesting question is whether this dependence can be made polynomial (or even removed entirely) with conditional samples.
One natural approach for this problem would be to consider it as an instance of equivalence testing: given $d$ independent samples from a distribution, one can form a single sample from the product of the marginals by taking $i$th coordinate from the $i$th sample.
We would then test whether this is equivalent to the original distribution. 
Unfortunately, it is not obvious how one could simulate \emph{conditional} samples from the product of the marginals.

Finally, the core problems studied in this paper have been on non-tolerant distribution testing: for instance, we want to test whether $p$ is \emph{exactly} uniform, or far from it.
One could generalize this by asking whether $p$ is $\varepsilon/2$-\emph{close} to uniform, or $\varepsilon$-far from it.
In the \SAMP model, tolerant testing is known to be much harder than non-tolerant testing~\cite{ValiantV17b, JiaoHW16,DaskalakisKW18}, increasing the complexity of identity testing from $\Theta(\sqrt{n})$ to $\Theta(n/\log n)$.
For the specific case of uniformity testing in the \COND model, surprisingly, both the tolerant and non-tolerant versions of the problem are $\Theta_\varepsilon(1)$~\cite{CanonneRS14}.
It is therefore natural to ask: what is the complexity of tolerant testing (of uniformity, identity, and equivalence) with conditional samples (\COND or \NACOND)?

\section*{Acknowledgments}
GK would like to thank Cl\'ement L. Canonne for many helpful discussions and suggesting the inclusion of Section~\ref{sec:standalone}, Adam Bouland for discussing the connection between conditional sampling and quantum postselection, and Sam Elder for suggesting the name \ANACONDA.
The authors would also like to thank the reviewers for numerous suggestions which have greatly improved the presentation of the paper.
\nocite{Mercier01, FiedorV12, ChellapillaF00, Heymann87, Anaconda2016, MarajJSCPR14}

\bibliographystyle{alpha}
\bibliography{biblio}
\appendix
\section{A Standalone Discrepancy-Finding Lemma}
\label{sec:standalone}
In this section, we prove a standalone variant of Lemma~\ref{lem:discrepant-set}.
The phrasing is chosen to be more general, so that it may be applicable for those who are working in settings besides the conditional sampling model for distribution testing.

\begin{lemma}
Let $p, q, z \in \mathbb{R}^n$ be vectors, where $z = p -q$ and $\|z\|_1 = \ve$.
Let $S$ be a (random) set generated by the following process: 
 pick $j \in [\log n]$ uniformly at random, let $r = 2^j$, and selects a random set $S \subseteq [n]$ by choosing each $i \in [n]$ to be in $S$ with probability $1/r$.
Then with probability $\Omega\left(\frac{1}{\log n}\right)$, the following occurs simultaneously:
\begin{itemize}
\item $|S| = \Theta\left(\frac{n}{r}\right)$;
\item There exists $i \in S$ such that $|z(i)| = \Omega(1/r)$.
\end{itemize}
\end{lemma}
\begin{proof}
This is very similar to the proof of Lemma~\ref{lem:discrepant-set}, so we only sketch the proof.
We will use the same $\Bin_j$ notation as introduced in Definition~\ref{def:bin}.

First, consider the case where $\sum_{j=\log(n/32)+1}{\log 5n} \sum_{i \in \Bin_j(|z|)} |z(i)| \geq 1/5$.
This case is similar to that in Section~\ref{sec:many-small}.
With $\Omega\left(\frac{1}{\log n}\right)$ probability, the sampling procedure will select $r = \log n$ and a set $S$ of size exactly $1$, so we condition on this event.
By a similar calculation as (\ref{eq:many-small}), in this case there are $\Omega(n)$ indices such that $|z(i)| = \Omega(1/r)$, and therefore there is an $\Omega(1)$ probability of selecting one.

Next, we consider the case where $\sum_{j=1}^{\log(n/32)} \sum_{i \in \Bin_j(|z|)} |z(i)| \geq 3/5$, which corresponds to Section~\ref{sec:few-small}.
By an analogue of Lemma~\ref{lem:rest-signal}, we have that $|S| = \Theta\left(\frac{n}{r}\right)$, satisfying the first condition.
Finally, an analogue of Lemma~\ref{lem:uniform-key}, we have that there exists an $i \in S$ which satisfies the second condition.
\end{proof}

\end{document}